\documentclass[sigconf,nonacm]{acmart}

\usepackage{amsmath,amsthm}
\usepackage{multirow,colortbl}
\usepackage{float}
\floatstyle{ruled}
\newfloat{algorithm}{tbp}{loa}
\floatname{algorithm}{Algorithm}
\floatstyle{plain}

\DeclareMathOperator{\KL}{KL}
\DeclareMathOperator{\CE}{CE}
\DeclareMathOperator{\sg}{sg}
\DeclareMathOperator{\TV}{TV}

\DeclareMathOperator{\uniq}{uniq}
\DeclareMathOperator{\norm}{norm}
\newcommand{\ind}[1]{\mathbf{1}\{#1\}}

\definecolor{OCBlue}{HTML}{356A9A}
\definecolor{OCTeal}{HTML}{187F77}
\definecolor{OCOrange}{HTML}{B45C22}
\definecolor{OCPurple}{HTML}{7C4B94}
\definecolor{OCTableBestBaseline}{HTML}{FFF0D6}
\definecolor{OCTableOurs}{HTML}{EAF2F8}
\definecolor{OCTableHeader}{HTML}{F2F5F8}
\definecolor{OCTableCategoryText}{HTML}{52616E}
\newsavebox{\ocmaintablemeasure}
\newcommand{\OCMainTableFit}[3]{%
  \begingroup
  \setlength{\tabcolsep}{0pt}%
  \sbox{\ocmaintablemeasure}{\begin{tabular}{#2}#3\end{tabular}}%
  \setlength{\tabcolsep}{\dimexpr
    (\linewidth-\wd\ocmaintablemeasure)/\numexpr 2*#1\relax\relax}%
  \begin{tabular}{#2}#3\end{tabular}%
  \endgroup
}
\newcommand{\OCMeanSD}[2]{\shortstack[c]{\textbf{#1}\\[-2pt]{\scriptsize$\pm$\,#2}}}
\newlength{\octermwidth}
\newlength{\oclabelwidth}
\newcommand{\octerm}[4][\displaystyle]{%
  \begingroup\setlength{\fboxsep}{1.2pt}%
  \settowidth{\octermwidth}{\colorbox{#2!10}{$#1 #3$}}%
  \settowidth{\oclabelwidth}{$\scriptstyle\text{\normalfont\bfseries #4}$}%
  \ifdim\oclabelwidth>\octermwidth
    \setlength{\octermwidth}{\oclabelwidth}%
  \fi
  \underbrace{\makebox[\octermwidth][c]{%
    \colorbox{#2!10}{$#1 #3$}}}%
  _{\textcolor{#2}{\text{\normalfont\bfseries #4}}}%
  \endgroup
}

\setcopyright{none}
\begin{document}

\title[OneChord]{OneChord: Distill What You Rank, Recommend in One Pass}
\author{Zesheng Li}
\affiliation{%
  \institution{University of Chinese Academy of Sciences}
  \city{Beijing}
  \country{China}}
\email{lizesheng26@mails.ucas.ac.cn}
\author{Chengchang Pan}
\affiliation{%
  \institution{University of Chinese Academy of Sciences}
  \city{Beijing}
  \country{China}}
\email{166353314@qq.com}
\author{Honggang Qi}
\authornote{Corresponding author.}
\affiliation{%
  \institution{University of Chinese Academy of Sciences}
  \city{Beijing}
  \country{China}}
\email{hgqi@ucas.ac.cn}
\renewcommand{\shortauthors}{Li et al.}

\begin{abstract}
Accurate product discovery on the Web requires recommendation models that rank catalog candidates within tight inference budgets. In semantic-ID recommendation, deployment can fuse code probabilities with item similarities, making the final ranking depend on their normalization and relative weights. This motivates a direct training principle: \emph{distill what you rank}. We present OneChord, which uses one learned, catalog-normalized fusion for next-item supervision, refinement distillation, and inference. A shared parallel decoder produces joint SID evidence and item-matching scores. During training, a frozen denoiser supplies legal partial SIDs; the student evaluates these states to construct refined fused targets. Matching direct and refined predictions on shared hard candidates transfers refinement into one-pass recommendation. We characterize the fusion weights induced by candidate conditioning and establish sufficient conditions for preserving the refined predictor's full-catalog Top-$K$ selection. Across six Amazon datasets, OneChord achieves the highest Recall@5 and NDCG@5 among the compared methods on every dataset, with average relative gains of 6.30\% and 7.12\% over the strongest baseline. It also achieves a 1.04--1.18$\times$ full-test generation speedup over the fastest baseline. OneChord makes the served ranking distribution the common object of recommendation learning and distillation.
\end{abstract}

\maketitle
\hypersetup{pdfauthor={Zesheng Li, Chengchang Pan, Honggang Qi}}

\section{Introduction}
\label{sec:introduction}

Product discovery on the Web depends on accurate recommendation under strict response-time budgets. Semantic identifiers (SIDs) represent catalog items as short code sequences~\cite{rajput2023tiger,wang2024letter}, and parallel decoding predicts their positions together~\cite{hou2025rpg,wang2026onegr,cai2026sprint}. Deployment can combine SID probabilities with item-matching evidence~\cite{tran2026epic,cai2026sprint}. The resulting list depends on the \emph{fused catalog distribution}: its support, normalization, and branch weights jointly determine which candidates reach the user. This makes the learning target a Web recommendation problem in its own right: \textbf{how can training-time refinement improve the ranking served by a one-pass predictor?}

\begin{figure}[t]
\centering
\includegraphics[width=\linewidth]{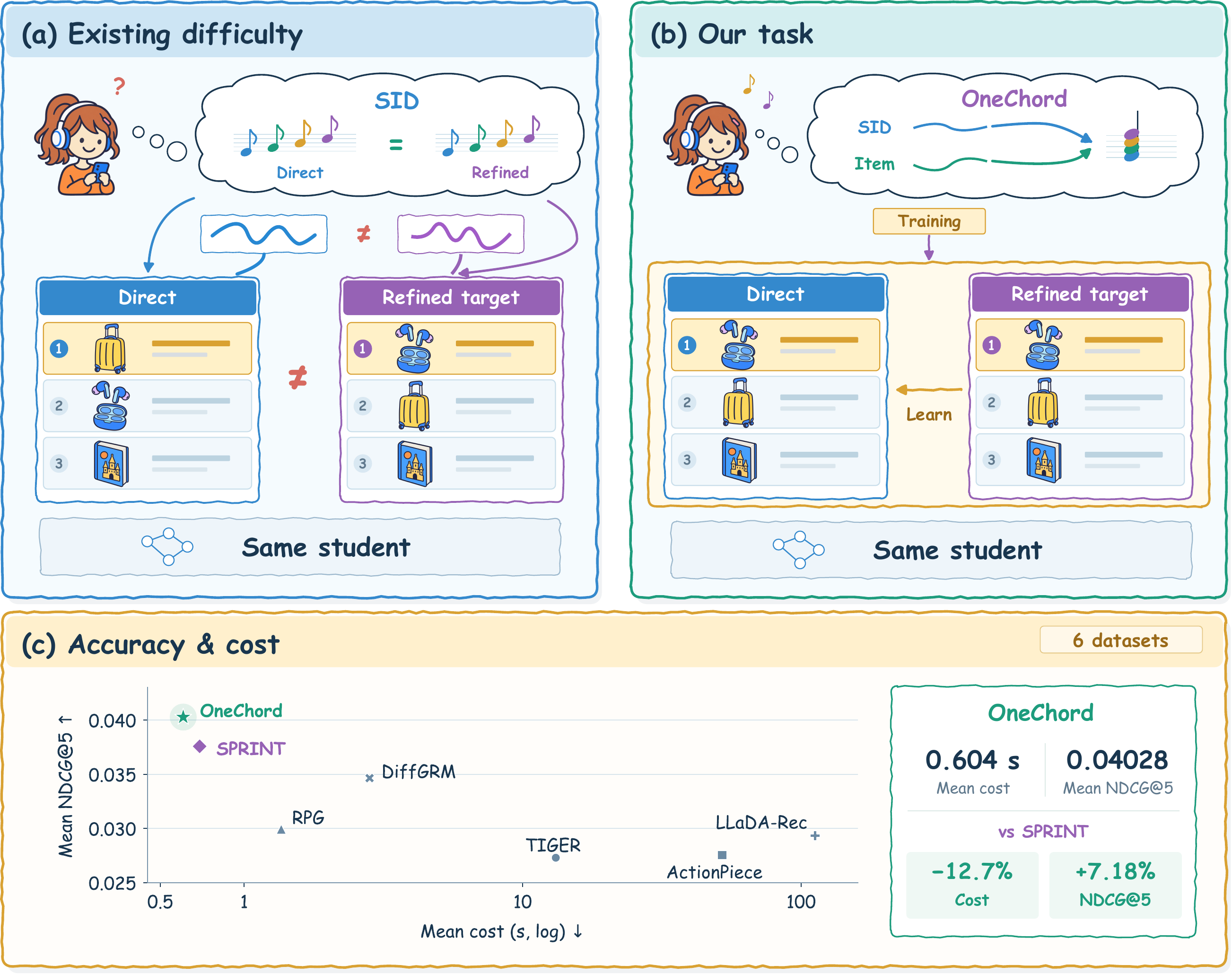}
\caption{\textbf{Distill what you rank, recommend in one pass.} (a) SID agreement can hide different item rankings. (b) OneChord aligns direct and refined fused rankings. (c) Mean NDCG@5 versus mean generation cost over six domains for all methods reported in both Tables~\ref{tab:main-results} and~\ref{tab:inference-cost}. The cost axis is logarithmic.}
\Description{Three panels show the motivation, approach and accuracy--cost comparison. Panel a shows identical SID predictions with different item rankings. Panel b aligns the direct fused ranking with a refined target from the same student. Panel c has one named point per method: TIGER, ActionPiece, RPG, LLaDA-Rec, DiffGRM, SPRINT and OneChord. The horizontal coordinate is the arithmetic mean cost in seconds over six domains, shown on a logarithmic axis. The vertical coordinate is the arithmetic mean NDCG at 5 over the same domains. Higher and further left is better.}
\label{fig:task}
\end{figure}

We introduce \textbf{OneChord}\footnote{Code: \url{https://anonymous.4open.science/r/OneChord-D1A2/}.}, built on \textbf{distill what you rank}. Next-item likelihood and ranking distillation optimize the same learned fusion used at inference. A shared encoder and parallel decoder produce joint SID probabilities and item similarities, with a discrete mixture coupling SID positions~\cite{hayakawa2025di4c}. A frozen denoiser supplies legal states that the student evaluates to construct refined fused targets. Catalog normalization precedes fusion and conditioning on shared hard candidates, preserving branch balance as refined rankings transfer into direct prediction (Figure~\ref{fig:task}).

Our analysis characterizes candidate-conditioned fusion and gives a margin condition for full-catalog Top-$K$ agreement. Across six Amazon domains, OneChord improves NDCG@5 over the best prior method on every domain, averaging \textbf{7.12\%} relative gain and reaching \textbf{9.96\%} on Beauty. It also generates recommendations \textbf{1.04--1.18$\times$} faster than the fastest baseline. Our contributions are:
\begin{list}{\textbullet}{%
  \setlength{\leftmargin}{0.8em}%
  \setlength{\labelwidth}{0.45em}%
  \setlength{\labelsep}{0.35em}%
  \setlength{\itemindent}{0pt}%
  \setlength{\listparindent}{0pt}%
  \setlength{\topsep}{2pt}%
  \setlength{\partopsep}{0pt}%
  \setlength{\itemsep}{0pt}%
  \setlength{\parsep}{0pt}%
}
\item \textbf{C1) Learn the served distribution.} A learned fusion unifies next-item supervision, refinement distillation, and inference.
\item \textbf{C2) Preserve the refined ranking.} Catalog-consistent targets and hard candidates connect distillation to Top-$K$ agreement.
\item \textbf{C3) Improve accuracy and efficiency.} OneChord gains 7.12\% in average NDCG@5 over the best prior method with faster one-pass generation.
\end{list}

\section{Related Work}
\label{sec:related_work}

\paragraph{\textbf{Generative recommendation and semantic indexing.}}
Attention over interaction histories models sequential preferences~\cite{kang2018sasrec,sun2019bert4rec}, while content and collaborative behavior provide complementary evidence for item representation and ranking~\cite{hou2022unisrec,hou2023vqrec,li2026bridge}. Generative approaches formulate recommendation as text or identifier prediction, as explored in P5 and TIGER~\cite{geng2022p5,rajput2023tiger}. Subsequent work integrates collaborative semantics, balances code usage, incorporates sequence context, or coordinates index and recommender training~\cite{zheng2024lcrec,wang2024letter,hou2025actionpiece,liu2025etegrec}. These directions enrich the information carried by semantic identifiers. OneChord addresses how SID and item evidence jointly determine recommendation: it learns and distills their \textbf{catalog-normalized fusion} over a shared semantic catalog.

\paragraph{\textbf{Parallel generation and catalog constraints.}}
Parallel decoders predict identifier positions together, with long SIDs and structured catalog search supporting a single neural pass~\cite{hou2025rpg,wang2026onegr}. Categorical corruption and denoising~\cite{austin2021d3pm,sahoo2024mdlm} support bidirectional SID modeling, adaptive masking, and refinement~\cite{shi2025lladarec,liu2026diffgrm,mu2026mdgr}. Item-aware prediction connects code evidence with catalog candidates~\cite{tran2026epic}. SPRINT~\cite{cai2026sprint} produces SID probabilities and item similarities in one forward pass, normalizes both over legal SIDs, and adds them for ranking. OneChord learns the fusion weight and directly supervises and distills \textbf{the final fused distribution}, retaining one-pass inference.

\paragraph{\textbf{Consistency and ranking distillation.}}
Recommendation distillation uses preferred items, relative orders, and latent representations to supervise compact predictors~\cite{tang2018rd,kang2020derrd}; lightweight decoder heads also reduce SID prediction cost~\cite{guo2026sidmlp}. Consistency models transfer iterative refinement into fewer evaluations~\cite{song2023consistency}. Di4C~\cite{hayakawa2025di4c} develops discrete mixture representations and boundary and composition objectives for dimensional dependence. OneChord uses these tools within a recommender whose primary transfer target is \textbf{the served fusion}: next-item likelihood calibrates its gate, and teacher-assisted states provide refined fused rankings on shared hard candidates. Catalog normalization precedes fusion and candidate restriction, preserving the deployed scoring rule during distillation.

\section{Method}
\label{sec:method}

OneChord learns \textbf{the distribution it serves}. Its design connects three operations: supervise a learned catalog fusion, construct refined predictions of that same fusion, and transfer their ranking to a direct predictor. Figure~\ref{fig:onechord} shows the shared predictor and its training objectives.

\begin{figure*}[t]
\centering
\includegraphics[width=0.8\textwidth]{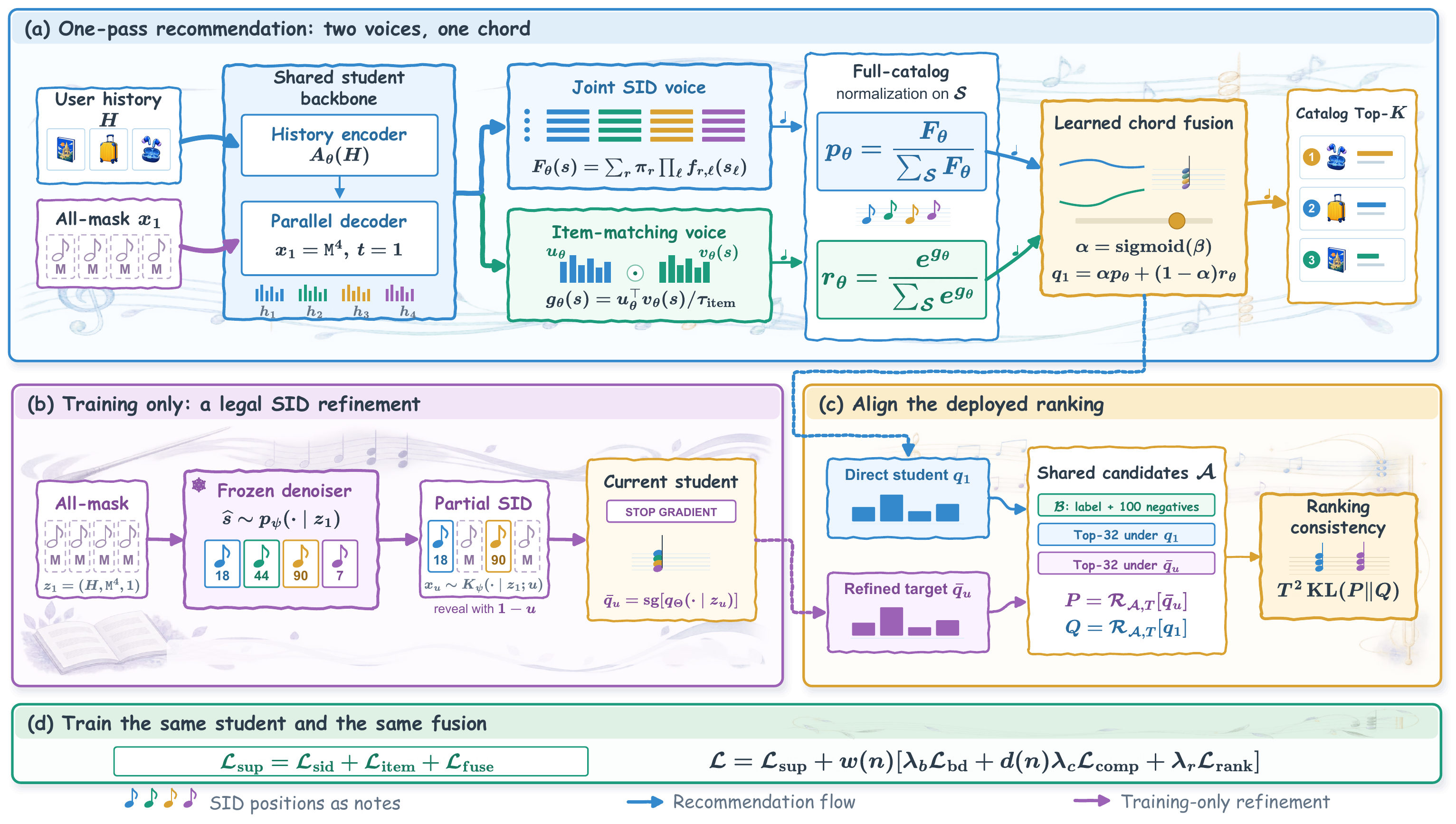}
\caption{\textbf{OneChord: distill what you rank, recommend in one pass.} (a) A shared history encoder and parallel decoder produce joint SID and item-matching evidence, normalized over the full catalog and combined through a learned fusion weight. (b) A frozen denoiser supplies a legal SID refinement, which the current student evaluates with stopped gradients. (c) Ranking consistency matches the direct and refined fused distributions on shared hard candidates. (d) Next-item supervision and scheduled distillation jointly train the same student and fusion.}
\Description{Four panels show OneChord. Panel a encodes the history, decodes all SID positions in one pass, normalizes joint SID and item-matching evidence on the full catalog, learns their fusion weight, and returns the top K. Panel b uses a frozen denoiser to refine an all-mask SID and evaluates the current student at the partial SID without target gradients. Panel c aligns direct and refined rankings on the union of supervised candidates and both distributions' top-scoring items. Panel d combines next-item supervision with scheduled boundary, composition, and ranking distillation.}
\label{fig:onechord}
\end{figure*}

\subsection{Semantic Catalog and Prediction States}
\label{sec:method-setup}

For an interaction history $H=(i_1,\ldots,i_m)$ and next item $i^+$, a fixed index assigns an $L$-code SID:
\begin{equation}
\begin{aligned}
\sigma(i)&=(s_1,\ldots,s_L)\in\textstyle\prod_{\ell=1}^{L}\mathcal V_\ell,\\
\mathcal S&=\{\sigma(i):i\in\mathcal I\},\qquad s^+=\sigma(i^+).
\end{aligned}
\label{eq:catalog-setup}
\end{equation}
The distinct legal SIDs $\mathcal S$ define the shared output space. The index is fixed before recommender training. At noise level $t\in[0,1]$, independent masking gives
\begin{equation}
Q_t(x\mid s)=\prod_{\ell=1}^{L}
\big[\octerm{OCTeal}{(1-t)\ind{x_\ell=s_\ell}}{Keep code}
+\octerm{OCPurple}{t\ind{x_\ell=\mathtt M}}{Mask code}\big].
\label{eq:masking}
\end{equation}
Write $z_t=(H,x_t,t)$. Partially visible states support refinement; the fully masked state $z_1=(H,\mathtt M^L,1)$ supports direct recommendation. Both use the same predictor.

\subsection{Shared One-Pass Catalog Predictor}
\label{sec:method-scoring}

\paragraph{\textbf{Shared decoding.}}
Code embeddings $E_\ell$ and a shared two-layer projection $M_\theta$ represent a history item as $M_\theta([E_1(s_1);\ldots;E_L(s_L)])$ plus a temporal position embedding. Bidirectional Transformer blocks~\cite{liu2026diffgrm,cai2026sprint} compute
\begin{equation}
\begin{aligned}
A_\theta(H)&=\operatorname{Enc}_\theta(H),\\
(h_1,\ldots,h_L)&=\operatorname{Dec}_\theta(A_\theta(H),x_t,t).
\end{aligned}
\label{eq:backbone}
\end{equation}
Decoder self-attention connects SID positions, and cross-attention reads history memory. All positions share one decoder evaluation and supply both scoring branches.

\paragraph{\textbf{Joint SID evidence.}}
A mixture of categorical products~\cite{hayakawa2025di4c} couples positions through a shared latent component:
\begin{equation}
F_\theta(s\mid z_t)=\sum_{r=1}^{R}
\octerm{OCBlue}{\pi_{\theta,r}(z_t)}{Component weight}
\octerm{OCTeal}{\prod_{\ell=1}^{L}f_{\theta,r,\ell}(s_\ell\mid z_t)}{Joint SID component}.
\label{eq:joint-mixture}
\end{equation}
Learned vectors $c_r$ specialize a shared decoder through
\begin{equation}
\begin{aligned}
a_{r,\ell}&=h_\ell+t c_r,\\
\widetilde h_{r,\ell}&=\operatorname{RMSNorm}
\big(a_{r,\ell}+B_\theta(a_{r,\ell})\big),\\
\pi_\theta(z_t)&=\operatorname{softmax}
\big(W_\pi[h_1;\ldots;h_L]+b_\pi\big).
\end{aligned}
\label{eq:component-parameterization}
\end{equation}
Here $B_\theta$ is a shared two-layer MLP. At masked positions, $f_{\theta,r,\ell}$ is the softmax of $\widetilde h_{r,\ell}^{\top}E_\ell(c)$ over codes occurring in the catalog. Visible codes are copied, giving $f_{\theta,r,\ell}(c\mid z_t)=\ind{c=x_{t,\ell}}$. Exact catalog conditioning yields
\begin{equation}
p_\theta(s\mid z_t)=\frac{F_\theta(s\mid z_t)}
{\octerm[\textstyle]{OCBlue}{\sum_{v\in\mathcal S}F_\theta(v\mid z_t)}{Legal-catalog mass}},
\qquad s\in\mathcal S.
\label{eq:sid-catalog}
\end{equation}

\paragraph{\textbf{Item-matching evidence.}}
With $\norm(a)=a/\lVert a\rVert_2$, the same hidden states produce
\begin{equation}
\begin{aligned}
u_\theta(z_t)&=\norm\!\big(M_\theta([h_1;\ldots;h_L])\big),\\
v_\theta(s)&=\norm\!\big(M_\theta(\sg[E_1(s_1);\ldots;E_L(s_L)])\big),\\
g_\theta(s\mid z_t)&=u_\theta(z_t)^\top v_\theta(s)/\tau_{\mathrm{item}},\\
r_\theta(s\mid z_t)&=\frac{\exp g_\theta(s\mid z_t)}
{\sum_{v\in\mathcal S}\exp g_\theta(v\mid z_t)}.
\end{aligned}
\label{eq:item-evidence}
\end{equation}
The stop-gradient operator $\sg$ detaches catalog-side code embeddings; the shared projection learns item matching. Both $p_\theta$ and $r_\theta$ are normalized over $\mathcal S$.

\subsection{Learning and Preserving the Served Fusion}
\label{sec:method-normalization}

\paragraph{\textbf{One distribution for supervision and ranking.}}
Let $\Theta=(\theta,\beta)$ and $\alpha=\operatorname{sigmoid}(\beta)$. A learned global gate defines the served distribution:
\begin{equation}
q_\Theta(s\mid z_t)=
\octerm{OCBlue}{\alpha p_\theta(s\mid z_t)}{SID evidence}
+\octerm{OCTeal}{(1-\alpha)r_\theta(s\mid z_t)}{Item evidence}.
\label{eq:fused-distribution}
\end{equation}
The next SID directly supervises this fusion at the fully masked endpoint:
\begin{equation}
\mathcal L_{\mathrm{fuse}}=-\mathbb E_{H,s^+}
\log q_\Theta(s^+\mid z_1).
\label{eq:fused-likelihood}
\end{equation}
Thus observed interactions train the branch balance and the same probabilities used to select recommendations.

\begin{proposition}[Supervised gate calibration]
\label{prop:gate-calibration}
For fixed branch probabilities $p_y,r_y>0$ at an observed SID, let $q_y=\alpha p_y+(1-\alpha)r_y$ and $\ell=-\log q_y$. Then
\begin{equation}
\frac{\partial\ell}{\partial\beta}=\alpha-\rho_y,
\qquad \rho_y=\octerm{OCBlue}{\frac{\alpha p_y}{q_y}}{SID responsibility}.
\label{eq:gate-gradient}
\end{equation}
The expected loss is convex in $\alpha$ with fixed branches.
\end{proposition}

\begin{proof}
Using $\partial\alpha/\partial\beta=\alpha(1-\alpha)$ gives
\begin{equation}
\frac{\partial\ell}{\partial\beta}
=\frac{\alpha(1-\alpha)(r_y-p_y)}{q_y}
=\alpha-\frac{\alpha p_y}{q_y}.
\end{equation}
For fixed branches,
\begin{equation}
\frac{\partial^2\ell}{\partial\alpha^2}
=\frac{(p_y-r_y)^2}{[\alpha p_y+(1-\alpha)r_y]^2}\geq0.
\end{equation}
Expectation preserves convexity in $\alpha$.
\end{proof}

The update calibrates the current gate against the SID branch's responsibility for the observed target. If the target receives more probability from the SID branch, its responsibility exceeds the current gate and gradient descent increases that branch's weight. The reverse update favors item evidence. This calibration is driven by observed interactions and is shared across users within a dataset. Ranking consistency additionally trains the same gate against refined fused predictions, as derived below.

\paragraph{\textbf{Preserving fusion on hard candidates.}}
Distillation concentrates on a candidate set $\mathcal A\subseteq\mathcal S$. For any distribution $a$, write $a(\mathcal A)=\sum_{s\in\mathcal A}a(s)$ and $a_{\mathcal A}(s)=a(s)/a(\mathcal A)$ on $\mathcal A$. Candidate conditioning induces a specific branch balance.

\begin{proposition}[Fusion under candidate conditioning]
\label{prop:conditioning}
Let $q=\alpha p+(1-\alpha)r$, with $0<\alpha<1$ and $p(\mathcal A),r(\mathcal A)>0$. Then
\begin{equation}
\begin{aligned}
q_{\mathcal A}&=\alpha_{\mathcal A}p_{\mathcal A}+(1-\alpha_{\mathcal A})r_{\mathcal A},\\
\alpha_{\mathcal A}&=\frac{\alpha p(\mathcal A)}
{\alpha p(\mathcal A)+(1-\alpha)r(\mathcal A)}.
\end{aligned}
\label{eq:conditional-gate}
\end{equation}
For $\widetilde q_{\mathcal A}=\alpha p_{\mathcal A}+(1-\alpha)r_{\mathcal A}$,
\begin{equation}
\TV(q_{\mathcal A},\widetilde q_{\mathcal A})=
\octerm{OCOrange}{|\alpha_{\mathcal A}-\alpha|}{Gate shift}
\octerm{OCTeal}{\TV(p_{\mathcal A},r_{\mathcal A})}{Branch gap}.
\label{eq:normalization-gap}
\end{equation}
\end{proposition}

\begin{proof}
Put $a=p(\mathcal A)$, $b=r(\mathcal A)$, and $Z=\alpha a+(1-\alpha)b$. On $\mathcal A$,
\begin{equation}
\frac{q(s)}{Z}=\frac{\alpha a}{Z}p_{\mathcal A}(s)
+\frac{(1-\alpha)b}{Z}r_{\mathcal A}(s).
\end{equation}
Subtracting $\widetilde q_{\mathcal A}$ gives $(\alpha_{\mathcal A}-\alpha)(p_{\mathcal A}-r_{\mathcal A})$. Taking half the $\ell_1$ norm proves the result.
\end{proof}

The ordering matters even for two candidates. With $p=(0.08,0.02,0.90)$, $r=(0.36,0.54,0.10)$, and $\alpha=1/2$, fusion ranks the second candidate ahead of the first. Conditioning each branch on the first two candidates and then equally fusing gives $(0.6,0.4)$, reversing that order. OneChord therefore \textbf{normalizes over the catalog, fuses, and then conditions}; its distillation target preserves the deployed scoring rule.

\subsection{Refining and Distilling Fused Rankings}
\label{sec:method-distillation}

\paragraph{\textbf{Legal refinement states.}}
A masked denoiser $D_\psi(s\mid z_t)$ uses one categorical component and visible-code copying. Its training objective is
\begin{equation}
\mathcal L_{\mathrm{teacher}}=
\mathbb E_{H,s^+,t,x_t\sim Q_t(\cdot\mid s^+)}
\left[-\frac{\log D_\psi(s^+\mid z_t)}{Lt}\right],
\label{eq:teacher-training}
\end{equation}
with $t$ drawn from a positive noise grid. After freezing $\psi$, catalog conditioning gives $p_\psi$. For $u<t$, sample one legal completion and reveal each masked code with probability $1-u/t$:
\begin{equation}
K_\psi(x_u\mid z_t;u)=\sum_{s\in\mathcal S}
\octerm{OCPurple}{p_\psi(s\mid z_t)}{Legal completion}
\octerm{OCTeal}{\prod_{\ell=1}^{L}\kappa_{t,u}(x_{u,\ell}\mid x_{t,\ell},s_\ell)}{Copy / reveal}.
\label{eq:teacher-kernel}
\end{equation}
The position kernel has the explicit form
\begin{equation}
\kappa_{t,u}(a\mid b,c)=
\begin{cases}
\ind{a=b},&b\neq\mathtt M,\\
\dfrac{u}{t}\ind{a=\mathtt M}
+\left(1-\dfrac{u}{t}\right)\ind{a=c},&b=\mathtt M.
\end{cases}
\label{eq:explicit-kernel}
\end{equation}
It copies visible codes and reveals masked positions from the same completion. If $\widehat s$ is sampled from the compatible catalog, every revealed code agrees with $\widehat s$; hence the refined state retains a legal completion. Appendix~\ref{app:teacher-kernel} extends this support argument to successive transitions.

\paragraph{\textbf{Refined targets in the served distribution.}}
Write $q_1=q_\Theta(\cdot\mid z_1)$. Starting from the fully masked state, sample $x_u\sim K_\psi(\cdot\mid z_1;u)$ and form $\bar q_u=\sg[q_\Theta(\cdot\mid H,x_u,u)]$. The denoiser supplies the state; the current student supplies the refined \emph{fused} ranking. Target evaluation reuses detached history memory. Its gradients are stopped, so the direct prediction learns to reproduce that ranking.

Let $\mathcal B$ contain $s^+$ and sampled catalog negatives. Denote by $\Top_b(a)$ the $b$ highest-scoring SIDs under $a$ among $\mathcal S\setminus\{s^+\}$. The shared candidate set is
\begin{equation}
\mathcal A=\uniq\!\big(
\octerm{OCTeal}{\mathcal B}{Supervised candidates}\cup
\octerm{OCOrange}{\Top_b(q_1)\cup\Top_b(\bar q_u)}{Hard union}\big).
\label{eq:candidate-union}
\end{equation}
For temperature $T>0$, post-fusion restriction is
\begin{equation}
\mathcal R_{\mathcal A,T}[q](s)=
\frac{\ind{s\in\mathcal A}\,q(s)^{1/T}}
{\sum_{v\in\mathcal A}q(v)^{1/T}}.
\label{eq:ranking-temperature}
\end{equation}
Ranking consistency matches the two resulting distributions:
\begin{equation}
\begin{gathered}
P=\octerm{OCPurple}{\mathcal R_{\mathcal A,T}[\bar q_u]}{Refined target},\qquad
Q=\octerm{OCBlue}{\mathcal R_{\mathcal A,T}[q_1]}{Direct prediction},\\
\mathcal L_{\mathrm{rank}}=T^2\mathbb E_{t=1}\KL(P\|Q).
\end{gathered}
\label{eq:ranking-distillation}
\end{equation}
Detached candidate selection keeps gradients on both scoring branches and their fusion gate. They learn on candidates competing for the final list.

\paragraph{\textbf{How ranking transfer calibrates fusion.}}
For fixed branches, candidates, and stopped target $P$, define the SID responsibility $\rho_s=\alpha p_\theta(s\mid z_1)/q_1(s)$. The per-example ranking loss has gate gradient
\begin{equation}
\frac{\partial\,[T^2\KL(P\|Q)]}{\partial\beta}
=T\sum_{s\in\mathcal A}(Q(s)-P(s))\rho_s.
\label{eq:ranking-gate-gradient}
\end{equation}
Thus ranking transfer calibrates the gate using branch responsibilities across the competing candidates, complementing the label responsibility in Eq.~\eqref{eq:gate-gradient}. Both signals update the fusion used at serving. The identity follows by differentiating the temperature-scaled conditional distribution; Appendix~\ref{app:gate-proof} gives the derivation.

\begin{proposition}[Ranking fidelity of the candidate union]
\label{prop:ranking-fidelity}
With $P,Q$ from Eq.~\eqref{eq:ranking-distillation}, let $\varepsilon=\KL(P\|Q)$ and let $A_P,A_Q\subset\mathcal A$ be their Top-$K$ sets. Then
\begin{equation}
P(A_P)-P(A_Q)\leq\sqrt{2\varepsilon}.
\label{eq:topk-regret}
\end{equation}
If $|\mathcal A|>K$ and the sorted target probabilities satisfy
\begin{equation}
\octerm{OCTeal}{P_{(K)}-P_{(K+1)}}{Target margin}
>\octerm{OCOrange}{\sqrt{2\varepsilon}}{Consistency error},
\label{eq:topk-margin}
\end{equation}
then $A_P=A_Q$. When $b\geq K$, the candidate union contains both full-catalog Top-$K$ sets, so the Top-$K$ sets of $\bar q_u$ and $q_1$ also coincide.
\end{proposition}

\begin{proof}
Since $A_Q$ maximizes $Q$-mass among size-$K$ candidate sets,
\begin{equation}
\begin{aligned}
P(A_P)-P(A_Q)
&\leq(P-Q)(A_P)-(P-Q)(A_Q)\\
&\leq2\TV(P,Q)\leq\sqrt{2\varepsilon}.
\end{aligned}
\end{equation}
Put $\gamma=P_{(K)}-P_{(K+1)}$. For $a\in A_P$ and $c\notin A_P$,
\begin{equation}
\begin{aligned}
Q(a)-Q(c)&\geq\gamma-|Q(a)-P(a)|-|Q(c)-P(c)|\\
&\geq\gamma-\lVert P-Q\rVert_1
\geq\gamma-\sqrt{2\varepsilon}>0.
\end{aligned}
\end{equation}
Therefore $A_P=A_Q$. Restriction and positive temperature scaling preserve within-set order. For $b\geq K$, each full-catalog Top-$K$ lies in its $b$ leading non-label candidates together with $s^+\in\mathcal B$. Both sets lie in $\mathcal A$, completing the argument.
\end{proof}

Here $\varepsilon$ is the per-prediction KL before the factor $T^2$. The result gives a selection guarantee through candidate agreement and coverage; Appendix~\ref{app:ranking-proof} relates it to the full-distribution KL.

% Queue the main result table ahead of the final training details.
% Accuracy is declared in the method; cost is declared on the next text page.
\begin{table*}[t]
\centering
\caption{Full-catalog test Recall/NDCG@5. OneChord: mean $\pm$ standard deviation. Amber: strongest baseline; blue: OneChord; bold: column best. Improv.\ is the relative gain over the strongest baseline.}
\label{tab:main-results}
\begingroup
\small
\renewcommand{\arraystretch}{1.10}
\OCMainTableFit{14}{ll*{12}{r}}{%
\toprule
& & \multicolumn{2}{c}{\textbf{Sports}} & \multicolumn{2}{c}{\textbf{Beauty}} & \multicolumn{2}{c}{\textbf{Toys}} & \multicolumn{2}{c}{\textbf{Scientific}} & \multicolumn{2}{c}{\textbf{Instruments}} & \multicolumn{2}{c}{\textbf{Games}} \\
\cmidrule(lr){3-4}\cmidrule(lr){5-6}\cmidrule(lr){7-8}\cmidrule(lr){9-10}\cmidrule(lr){11-12}\cmidrule(lr){13-14}
\rowcolor{OCTableHeader}
\textbf{Type} & \textbf{Method} & R@5 & N@5 & R@5 & N@5 & R@5 & N@5 & R@5 & N@5 & R@5 & N@5 & R@5 & N@5 \\
\midrule
\multirow{5}{*}{\textcolor{OCTableCategoryText}{\shortstack[l]{Item\\ID-based}}} & GRU4Rec~\cite{hidasi2016gru4rec}~(2016) & 0.0128 & 0.0087 & 0.0165 & 0.0098 & 0.0098 & 0.0060 & 0.0186 & 0.0127 & 0.0293 & 0.0195 & 0.0459 & 0.0310 \\
 & SASRec~\cite{kang2018sasrec}~(2018) & 0.0232 & 0.0155 & 0.0390 & 0.0248 & 0.0462 & 0.0304 & 0.0237 & 0.0150 & 0.0335 & 0.0212 & 0.0512 & 0.0326 \\
 & BERT4Rec~\cite{sun2019bert4rec}~(2019) & 0.0116 & 0.0074 & 0.0202 & 0.0125 & 0.0117 & 0.0070 & 0.0159 & 0.0101 & 0.0252 & 0.0158 & 0.0311 & 0.0198 \\
 & PreferDiff~\cite{liu2025preferdiff}~(2025) & 0.0276 & 0.0189 & 0.0453 & 0.0318 & 0.0605 & 0.0406 & 0.0163 & 0.0115 & 0.0210 & 0.0140 & 0.0353 & 0.0244 \\
 & FAVE~\cite{shi2026fave}~(2026) & 0.0303 & 0.0193 & 0.0583 & 0.0407 & 0.0615 & 0.0447 & 0.0229 & 0.0144 & 0.0324 & 0.0208 & 0.0544 & 0.0350 \\
\midrule
\multirow{4}{*}{\textcolor{OCTableCategoryText}{\shortstack[l]{SID-based\\(AR)}}} & TIGER~\cite{rajput2023tiger}~(2023) & 0.0267 & 0.0179 & 0.0455 & 0.0325 & 0.0519 & 0.0373 & 0.0279 & 0.0182 & 0.0360 & 0.0234 & 0.0528 & 0.0345 \\
 & ActionPiece~\cite{hou2025actionpiece}~(2025) & 0.0317 & 0.0203 & 0.0512 & 0.0341 & 0.0473 & 0.0311 & 0.0269 & 0.0172 & 0.0390 & 0.0245 & 0.0583 & 0.0381 \\
 & COBRA~\cite{yang2025cobra}~(2025) & 0.0301 & 0.0216 & 0.0539 & 0.0394 & 0.0621 & 0.0461 & 0.0275 & 0.0178 & 0.0345 & 0.0229 & 0.0473 & 0.0309 \\
 & SID-MLP~\cite{guo2026sidmlp}~(2026) & 0.0263 & 0.0174 & 0.0442 & 0.0303 & 0.0378 & 0.0243 & 0.0292 & 0.0194 & 0.0398 & 0.0254 & 0.0575 & 0.0381 \\
\midrule
\multirow{5}{*}{\textcolor{OCTableCategoryText}{\shortstack[l]{SID-based\\(NAR)}}} & RPG~\cite{hou2025rpg}~(2025) & 0.0312 & 0.0214 & 0.0546 & 0.0377 & 0.0593 & 0.0397 & 0.0259 & 0.0172 & 0.0358 & 0.0242 & 0.0582 & 0.0394 \\
 & LLaDA-Rec~\cite{shi2025lladarec}~(2025) & 0.0280 & 0.0186 & 0.0508 & 0.0348 & 0.0509 & 0.0345 & 0.0308 & 0.0202 & 0.0403 & 0.0265 & 0.0624 & 0.0414 \\
 & DiffGRM~\cite{liu2026diffgrm}~(2026) & 0.0366 & 0.0248 & 0.0605 & 0.0413 & 0.0617 & 0.0453 & 0.0404 & 0.0276 & 0.0430 & 0.0284 & 0.0619 & 0.0405 \\
 & OneGR~\cite{wang2026onegr}~(2026) & 0.0375 & 0.0254 & 0.0632 & 0.0430 & 0.0632 & 0.0462 & 0.0353 & 0.0244 & 0.0367 & 0.0248 & 0.0545 & 0.0366 \\
 & SPRINT~\cite{cai2026sprint}~(2026) & \cellcolor{OCTableBestBaseline}0.0413 & \cellcolor{OCTableBestBaseline}0.0271 & \cellcolor{OCTableBestBaseline}0.0674 & \cellcolor{OCTableBestBaseline}0.0452 & \cellcolor{OCTableBestBaseline}0.0705 & \cellcolor{OCTableBestBaseline}0.0500 & \cellcolor{OCTableBestBaseline}0.0442 & \cellcolor{OCTableBestBaseline}0.0300 & \cellcolor{OCTableBestBaseline}0.0460 & \cellcolor{OCTableBestBaseline}0.0306 & \cellcolor{OCTableBestBaseline}0.0635 & \cellcolor{OCTableBestBaseline}0.0426 \\
\midrule
\rowcolor{OCTableOurs}
\textcolor{OCTableCategoryText}{Ours} & \textbf{OneChord} & \OCMeanSD{0.0438}{0.0008} & \OCMeanSD{0.0295}{0.0006} & \OCMeanSD{0.0731}{0.0010} & \OCMeanSD{0.0497}{0.0008} & \OCMeanSD{0.0746}{0.0011} & \OCMeanSD{0.0546}{0.0009} & \OCMeanSD{0.0475}{0.0009} & \OCMeanSD{0.0325}{0.0007} & \OCMeanSD{0.0480}{0.0008} & \OCMeanSD{0.0319}{0.0006} & \OCMeanSD{0.0671}{0.0007} & \OCMeanSD{0.0435}{0.0005} \\
\addlinespace[1pt]
 & \textit{Improv.} & 6.05\% & 8.86\% & 8.46\% & 9.96\% & 5.82\% & 9.20\% & 7.47\% & 8.33\% & 4.35\% & 4.25\% & 5.67\% & 2.11\% \\
\bottomrule
}
\endgroup

\end{table*}

\subsection{SID Consistency across Refinement States}
\label{sec:method-sid-consistency}

Boundary and composition objectives~\cite{hayakawa2025di4c} additionally train the SID branch across noise levels. For $D$ grid intervals and $\delta=1/D$, sample $x_\delta\sim Q_\delta(\cdot\mid s^+)$. Boundary matching uses
\begin{equation}
\mathcal L_{\mathrm{bd}}=\frac{1}{L}\mathbb E
\KL\!\big(\octerm{OCPurple}{p_\psi(\cdot\mid z_\delta)}{Frozen denoiser}\,\big\|\,
\octerm{OCBlue}{p_\theta(\cdot\mid z_\delta)}{Student}\big).
\label{eq:boundary}
\end{equation}
For $u=t-\delta$, let $\bar p_\theta$ denote the stopped student and $\mathcal X=\prod_\ell(\mathcal V_\ell\cup\{\mathtt M\})$. Its averaged refinement is
\begin{equation}
\begin{aligned}
m_{t,u}(s\mid z_t)&=(\mathcal T_\psi^{t,u}\bar p_\theta)(s\mid z_t)\\
&=\int_{\mathcal X}\bar p_\theta(s\mid H,x,u)\,
K_\psi(\mathrm dx\mid z_t;u).
\end{aligned}
\label{eq:refinement-average}
\end{equation}
On this finite space, the integral is a sum over refinement outcomes. With $\nu_D$ the grid sampling measure, composition consistency is
\begin{equation}
\begin{aligned}
\mathcal L_{\mathrm{comp}}&=\frac{1}{L}\mathbb E_{H,s^+}
\int_{(\delta,1]}\mathbb E_{x_t\sim Q_t(\cdot\mid s^+)}\\
&\quad\CE\!\big(\octerm{OCPurple}{m_{t,u}(\cdot\mid z_t)}{Refinement average},
\octerm{OCBlue}{p_\theta(\cdot\mid z_t)}{Source prediction}\big)\,\nu_D(\mathrm dt).
\end{aligned}
\label{eq:composition}
\end{equation}
Linearity in the target gives the sampled objective:
\begin{equation}
\begin{aligned}
&\mathbb E_{x_u\sim K_\psi}
\CE(\bar p_\theta(\cdot\mid z_u),p_\theta(\cdot\mid z_t))\\
&\qquad=\CE(m_{t,u},p_\theta(\cdot\mid z_t)).
\end{aligned}
\label{eq:sampled-composition}
\end{equation}
A single refined state supplies an unbiased estimate of the composition objective and its online gradient, conditional on the source state and stopped target generation. The expectation averages across legal reveal patterns, while each sampled target remains a distribution over the catalog. This permits stochastic training of the refinement operator through ordinary cross-entropy without enumerating the space of masked states. The teacher kernel determines which states provide the consistency signal.

For $a_x=\bar p_\theta(\cdot\mid H,x,u)$, $m=\mathbb E_xa_x$, and $b=p_\theta(\cdot\mid z_t)$, the KL decomposition further yields
\begin{equation}
\begin{aligned}
\mathbb E_x\KL(a_x\|b)
&=\octerm{OCBlue}{\KL(m\|b)}{Direct prediction error}\\
&\quad+\octerm{OCPurple}{\mathbb E_x\KL(a_x\|m)}{Refinement variation}.
\end{aligned}
\label{eq:refinement-kl-decomposition}
\end{equation}
With stopped targets, the second term is constant for the online update. Composition therefore trains the direct SID prediction toward the refinement average. Appendix~\ref{app:composition} derives these identities and the exact finite-grid measure.

\subsection{Optimization and One-Pass Recommendation}
\label{sec:method-training}

\paragraph{\textbf{Joint objective.}}
The direct branch losses use the same supervised multiset $\mathcal B$:
\begin{equation}
\begin{aligned}
\mathcal L_{\mathrm{sid}}&=-\mathbb E\log
\frac{F_\theta(s^+\mid z_1)}{\sum_{v\in\mathcal B}F_\theta(v\mid z_1)},\\
\mathcal L_{\mathrm{item}}&=-\mathbb E\log
\frac{\exp g_\theta(s^+\mid z_1)}{\sum_{v\in\mathcal B}\exp g_\theta(v\mid z_1)},\\
\mathcal L_{\mathrm{sup}}&=\lambda_s\mathcal L_{\mathrm{sid}}+
\lambda_i\mathcal L_{\mathrm{item}}+\lambda_f\mathcal L_{\mathrm{fuse}}.
\end{aligned}
\label{eq:supervised-losses}
\end{equation}
The fused likelihood uses the complete catalog; branch losses retain sampling multiplicity. At step $n$, the student minimizes
\begin{equation}
\begin{aligned}
\mathcal L(n)&=\octerm{OCTeal}{\mathcal L_{\mathrm{sup}}}{Next-item supervision}\\
&\quad+w(n)\big[\octerm{OCPurple}{\lambda_b\mathcal L_{\mathrm{bd}}}{Boundary}
+\octerm{OCBlue}{d(n)\lambda_c\mathcal L_{\mathrm{comp}}}{SID consistency}\\
&\hspace{6.2em}+\octerm{OCOrange}{\lambda_r\mathcal L_{\mathrm{rank}}}{Fused ranking}\big].
\end{aligned}
\label{eq:total-loss}
\end{equation}
The ramp $w(n)$ establishes next-item prediction before increasing distillation; $d(n)$ later reduces SID composition while maintaining fused ranking transfer. The student starts from an independent initialization. Algorithm~\ref{alg:onechord} summarizes training and serving; Appendix~\ref{app:implementation} specifies the schedules and coefficients.

\begin{table*}[t]
\centering
\caption{Generation time (s) over the complete test set and the speedup of OneChord over each baseline.}
\label{tab:inference-cost}
\begingroup
\small
\renewcommand{\arraystretch}{1.07}
\OCMainTableFit{13}{l*{12}{r}}{%
\toprule
& \multicolumn{2}{c}{\textbf{Sports}} & \multicolumn{2}{c}{\textbf{Beauty}} & \multicolumn{2}{c}{\textbf{Toys}} & \multicolumn{2}{c}{\textbf{Scientific}} & \multicolumn{2}{c}{\textbf{Instruments}} & \multicolumn{2}{c}{\textbf{Games}} \\
\cmidrule(lr){2-3}\cmidrule(lr){4-5}\cmidrule(lr){6-7}\cmidrule(lr){8-9}\cmidrule(lr){10-11}\cmidrule(lr){12-13}
\rowcolor{OCTableHeader}
\textbf{Method} & Cost & Speedup & Cost & Speedup & Cost & Speedup & Cost & Speedup & Cost & Speedup & Cost & Speedup \\
\midrule
TIGER~\cite{rajput2023tiger}~(2023) & 10.0199 & 22.65$\times$ & 6.2677 & 22.94$\times$ & 5.4396 & 23.28$\times$ & 14.3327 & 21.39$\times$ & 16.1564 & 21.18$\times$ & 26.7062 & 21.47$\times$ \\
ActionPiece~\cite{hou2025actionpiece}~(2025) & 39.2035 & 88.64$\times$ & 24.7523 & 90.60$\times$ & 21.4593 & 91.82$\times$ & 56.9978 & 85.08$\times$ & 63.8786 & 83.74$\times$ & 105.9439 & 85.16$\times$ \\
RPG~\cite{hou2025rpg}~(2025) & 0.7240 & 1.64$\times$ & 0.5641 & 2.06$\times$ & 0.3615 & 1.55$\times$ & 1.6441 & 2.45$\times$ & 1.8188 & 2.38$\times$ & 3.0426 & 2.45$\times$ \\
LLaDA-Rec~\cite{shi2025lladarec}~(2025) & 85.1093 & 192.43$\times$ & 56.8194 & 207.96$\times$ & 45.6964 & 195.54$\times$ & 121.7687 & 181.76$\times$ & 136.7625 & 179.28$\times$ & 228.0063 & 183.27$\times$ \\
DiffGRM~\cite{liu2026diffgrm}~(2026) & 2.1590 & 4.88$\times$ & 1.3601 & 4.98$\times$ & 1.1404 & 4.88$\times$ & 3.0965 & 4.62$\times$ & 3.4338 & 4.50$\times$ & 5.7371 & 4.61$\times$ \\
SPRINT~\cite{cai2026sprint}~(2026) & \cellcolor{OCTableBestBaseline}0.5007 & 1.13$\times$ & \cellcolor{OCTableBestBaseline}0.2844 & 1.04$\times$ & \cellcolor{OCTableBestBaseline}0.2477 & 1.06$\times$ & \cellcolor{OCTableBestBaseline}0.7681 & 1.15$\times$ & \cellcolor{OCTableBestBaseline}0.8861 & 1.16$\times$ & \cellcolor{OCTableBestBaseline}1.4658 & 1.18$\times$ \\
\midrule
\rowcolor{OCTableOurs}
\textbf{OneChord} & \textbf{0.4423} & --- & \textbf{0.2732} & --- & \textbf{0.2337} & --- & \textbf{0.6700} & --- & \textbf{0.7629} & --- & \textbf{1.2441} & --- \\
\bottomrule
}
\endgroup

\end{table*}

\paragraph{\textbf{How fused supervision trains the branches.}}
The observed target calibrates both the gate and the evidence combined by that gate. With $y=s^+$, $p_y=p_\theta(y\mid z_1)$, $r_y=r_\theta(y\mid z_1)$, and $\rho_y=\alpha p_y/q_y$, the fused likelihood gives
\begin{equation}
\begin{aligned}
\nabla_\theta[-\log q_y]
&=\octerm{OCBlue}{\rho_y\nabla_\theta[-\log p_y]}{SID contribution}\\
&\quad+\octerm{OCTeal}{(1-\rho_y)\nabla_\theta[-\log r_y]}{Item contribution}.
\end{aligned}
\label{eq:branch-responsibility-gradient}
\end{equation}
This identity holds with shared parameters: each branch contributes through the same history encoder and decoder. Its target responsibility weights the branch gradient, while Eq.~\eqref{eq:gate-gradient} updates the scalar fusion logit. The branch objectives additionally sharpen SID and item discrimination on sampled negatives. Together, these signals connect representation learning to the probability assigned by the final fusion.

Ranking distillation supplies a complementary signal over competing candidates. Its soft target distributes probability across the shared hard union, and Eq.~\eqref{eq:ranking-gate-gradient} aggregates their branch responsibilities. The observed next SID thus anchors prediction, while refined rankings train the relative allocation of probability around that target. The same $q_\Theta$ receives both updates and determines the recommendation list.

\begin{algorithm}[t]
\caption{OneChord: learn and distill the served distribution}
\label{alg:onechord}
\small
\raggedright
\setlength{\tabcolsep}{0pt}
\renewcommand{\arraystretch}{1.08}
\begin{tabular}{@{}>{\color{OCTableCategoryText}\footnotesize}r@{\hspace{0.7em}}>{\raggedright\arraybackslash}p{\dimexpr\linewidth-1.7em\relax}@{}}
\multicolumn{2}{@{}l@{}}{\textbf{Input}\enspace Catalog $\mathcal S$; frozen denoiser $p_\psi$; student $\Theta$.}\\[3pt]
\rowcolor{OCTableOurs}
\multicolumn{2}{@{}p{\linewidth}@{}}{\strut\hspace{3pt}\textcolor{OCBlue}{\textbf{Training}}\hfill\textit{Minibatch $(H,s^+)$ at step $n$}\hspace*{3pt}}\\[2pt]
1 & \textbf{Predict:} encode $H$, decode $z_1$, and form $q_1$ over $\mathcal S$.\\[1pt]
2 & \textbf{Supervise:} compute branch losses and fused likelihood.\\[1pt]
3 & \textbf{Refine:} sample noise, masks, and legal states using $K_\psi$.\\[1pt]
4 & \textbf{Match SIDs:} compute $\mathcal L_{\mathrm{bd}},\mathcal L_{\mathrm{comp}}$ with stopped targets.\\[1pt]
5 & For $t=1$, evaluate stopped refined fusion $\bar q_u$ and build the hard union $\mathcal A$ from $q_1,\bar q_u,\mathcal B$.\\[1pt]
6 & \textbf{Match rankings:} restrict $q_1,\bar q_u$ with $\mathcal R_{\mathcal A,T}$ and compute $\mathcal L_{\mathrm{rank}}$ on fully masked sources.\\[1pt]
7 & \textbf{Update} $\Theta$ with the scheduled objective $\mathcal L(n)$.\\[3pt]
\rowcolor{OCTableOurs}
\multicolumn{2}{@{}p{\linewidth}@{}}{\strut\hspace{3pt}\textcolor{OCBlue}{\textbf{Serving}}\hfill\textit{One forward pass}\hspace*{3pt}}\\[2pt]
8 & Encode $H$ once, decode $z_1$ once, and form $q_1$ over $\mathcal S$.\\[1pt]
9 & \textbf{return} $\Top_K\{\log q_1(s):s\in\mathcal S\}$.\\
\end{tabular}
\end{algorithm}

\paragraph{\textbf{Serving.}}
OneChord evaluates its shared encoder and decoder once at $z_1$ and returns
\begin{equation}
\widehat{\mathcal S}_K(H)=\Top_K\{\log q_\Theta(s\mid z_1):s\in\mathcal S\}.
\label{eq:inference}
\end{equation}
Cached vectors $v_\theta(s)$ give $O(|\mathcal S|d)$ item matching for width $d$. Joint SID scoring costs $O(|\mathcal S|RL)$ after token prediction, followed by Top-$K$ selection on the fused scores. Refinement is transferred during training; serving uses one history encoding, one parallel decoding, and catalog scoring.

\paragraph{\textbf{Shared computation over the catalog.}}
The decoder produces a token-probability table for every component and position. Every legal SID reuses these tables; its joint score is evaluated as
\begin{equation}
\log F_\theta(s\mid z_1)=\operatorname{LSE}_{r}
\left[\log\pi_{\theta,r}+\sum_{\ell=1}^{L}\log f_{\theta,r,\ell}(s_\ell)\right],
\label{eq:catalog-log-mixture}
\end{equation}
where $\operatorname{LSE}$ denotes log-sum-exp. The item branch similarly reuses one decoder-derived user vector against cached catalog vectors. Both normalizers are computed over the same $\mathcal S$, after which the learned gate supplies the final scores. Increasing the catalog enlarges these shared scoring operations while retaining one encoder evaluation and one parallel decoder evaluation. This separates the neural prediction budget from catalog coverage and permits exact legal-SID ranking within the one-pass serving path.

\section{Experiments}
\label{sec:experiments}

The experiments establish accuracy and generation efficiency, examine supervision and distillation of the served fusion, and characterize its behavior across branch weights and interaction sparsity.

\subsection{Experimental Setup}
\label{sec:experimental-setup}

\paragraph{\textbf{Datasets and Evaluation.}}
We evaluate six recommendation datasets: Sports, Beauty, and Toys from Amazon Reviews 2014~\cite{mcauley2015amazon}, and Scientific, Instruments, and Games from Amazon Reviews 2023~\cite{hou2026amazon}. Table~\ref{tab:dataset-statistics} gives their statistics. We report Recall@5 (R@5) and NDCG@5 (N@5) from rankings over the complete legal SID catalog~\cite{hou2025rpg,cai2026sprint}; a hit requires all four predicted codes to match the held-out target SID. Generation cost is the total time over the complete test set. Data splits, model selection, and run settings are specified in Appendix~\ref{app:experimental-protocol}.

\paragraph{\textbf{Baselines.}}
We compare with $14$ baselines spanning item ID-based, autoregressive SID-based, and non-autoregressive SID-based recommendation. Table~\ref{tab:main-results} lists the methods and their references.

\paragraph{\textbf{Implementation Details.}}
OneChord uses a shared OPQ catalog and a bidirectional Transformer backbone, with parallel SID prediction and learnable SID--item fusion. Detailed architecture, training configurations, evaluation procedures, and reproducibility settings are provided in Appendices~\ref{app:implementation} and~\ref{app:experimental-protocol}.

\subsection{Overall Performance and Efficiency}
\label{sec:overall-performance}

Table~\ref{tab:main-results} shows that OneChord achieves the highest Recall@5 and NDCG@5 on all six datasets. Relative to the strongest baseline in each column, the gains average \textbf{6.30\%} in Recall@5 and \textbf{7.12\%} in NDCG@5, reaching \textbf{9.96\%} in NDCG@5 on Beauty. These improvements span all three baseline families while retaining one parallel decoding pass at inference.

The gains cover both target retrieval and the quality of its position in the list. On Sports, Beauty, Toys, and Scientific, relative NDCG@5 improvements are $8.86\%$, $9.96\%$, $9.20\%$, and $8.33\%$, respectively; the corresponding Recall@5 gains range from $5.82\%$ to $8.46\%$. Instruments and Games also improve both metrics, extending the result across all six domains. These paired gains support learning a distribution whose probability allocation serves the final ranked list. The cross-domain summaries average the six relative improvements, giving each domain equal weight.

\paragraph{\textbf{Inference Efficiency.}}
\label{sec:inference-efficiency}

OneChord completes full-test Top-10 generation in \textbf{0.23--1.24 seconds}, achieving \textbf{1.04--1.18$\times$} speedup over the fastest baseline across all six domains (Table~\ref{tab:inference-cost}). It combines the best accuracy in Table~\ref{tab:main-results} with the lowest generation cost in Table~\ref{tab:inference-cost}. Its serving path uses one history encoding, one parallel decoding, and catalog scoring.

The accuracy and cost tables evaluate complementary aspects of the same recommendation pipeline. The first measures the positions of held-out targets in complete-catalog rankings; the second measures generation of the final Top-10 list. OneChord's learned fusion participates in both. Its refinement objectives enrich training, while the deployed predictor retains the direct shared-backbone computation described in Section~\ref{sec:method-training}. Figure~\ref{fig:task} presents the mean accuracy--cost relationship across domains.

% Generated by paper/tables/make_ablation_table.py from saved test results.
\begin{table*}[t]
\caption{Full-catalog test ablations over six datasets. Blue: full OneChord. The controlled protocol is specified in Appendix~\ref{app:ablation-protocol}.}
\label{tab:ablation}
\centering
\begingroup
\small
\renewcommand{\arraystretch}{1.12}
\OCMainTableFit{13}{l*{12}{r}}{%
\toprule
 & \multicolumn{2}{c}{\textbf{Sports}} & \multicolumn{2}{c}{\textbf{Beauty}} & \multicolumn{2}{c}{\textbf{Toys}} & \multicolumn{2}{c}{\textbf{Scientific}} & \multicolumn{2}{c}{\textbf{Instruments}} & \multicolumn{2}{c}{\textbf{Games}} \\
\cmidrule(lr){2-3}\cmidrule(lr){4-5}\cmidrule(lr){6-7}\cmidrule(lr){8-9}\cmidrule(lr){10-11}\cmidrule(lr){12-13}
\rowcolor{OCTableHeader}
\textbf{Variant} & R@5 & N@5 & R@5 & N@5 & R@5 & N@5 & R@5 & N@5 & R@5 & N@5 & R@5 & N@5 \\
\midrule
\rowcolor{OCTableOurs}
\textbf{OneChord} & 0.0438 & 0.0295 & 0.0731 & 0.0497 & 0.0746 & 0.0546 & 0.0475 & 0.0325 & 0.0480 & 0.0319 & 0.0671 & 0.0435 \\
\midrule
Fixed equal fusion ($\alpha=1/2$) & 0.0440 & 0.0293 & 0.0699 & 0.0490 & 0.0747 & 0.0541 & 0.0463 & 0.0321 & 0.0466 & 0.0308 & 0.0657 & 0.0434 \\
w/o fused likelihood & 0.0405 & 0.0269 & 0.0663 & 0.0448 & 0.0717 & 0.0519 & 0.0425 & 0.0294 & 0.0413 & 0.0278 & 0.0585 & 0.0391 \\
Single SID component & 0.0429 & 0.0283 & 0.0695 & 0.0485 & 0.0741 & 0.0537 & 0.0477 & 0.0328 & 0.0487 & 0.0322 & 0.0665 & 0.0436 \\
\midrule
w/o all distillation & 0.0419 & 0.0278 & 0.0702 & 0.0476 & 0.0716 & 0.0518 & 0.0452 & 0.0305 & 0.0455 & 0.0297 & 0.0642 & 0.0413 \\
w/o SID distillation & 0.0427 & 0.0285 & 0.0717 & 0.0485 & 0.0730 & 0.0532 & 0.0463 & 0.0316 & 0.0466 & 0.0308 & 0.0656 & 0.0425 \\
w/o boundary distillation & 0.0429 & 0.0287 & 0.0720 & 0.0488 & 0.0732 & 0.0536 & 0.0467 & 0.0319 & 0.0469 & 0.0311 & 0.0659 & 0.0428 \\
w/o composition consistency & 0.0431 & 0.0288 & 0.0718 & 0.0486 & 0.0735 & 0.0537 & 0.0465 & 0.0317 & 0.0470 & 0.0312 & 0.0660 & 0.0427 \\
w/o ranking distillation & 0.0426 & 0.0284 & 0.0715 & 0.0483 & 0.0728 & 0.0531 & 0.0461 & 0.0315 & 0.0465 & 0.0307 & 0.0654 & 0.0423 \\
Item-only ranking distillation & 0.0430 & 0.0289 & 0.0721 & 0.0489 & 0.0736 & 0.0538 & 0.0468 & 0.0319 & 0.0468 & 0.0309 & 0.0661 & 0.0429 \\
\midrule
w/o catalog normalization & 0.0424 & 0.0287 & 0.0700 & 0.0485 & 0.0742 & 0.0545 & 0.0473 & 0.0326 & 0.0476 & 0.0317 & 0.0662 & 0.0437 \\
w/o hard candidates & 0.0430 & 0.0283 & 0.0705 & 0.0486 & 0.0745 & 0.0540 & 0.0475 & 0.0327 & 0.0478 & 0.0321 & 0.0658 & 0.0431 \\
w/o curriculum & 0.0421 & 0.0279 & 0.0728 & 0.0502 & 0.0745 & 0.0541 & 0.0455 & 0.0317 & 0.0481 & 0.0321 & 0.0669 & 0.0439 \\
\bottomrule
}
\endgroup
\end{table*}

\begin{figure}[!htbp]
\centering
\includegraphics[width=0.8\columnwidth]{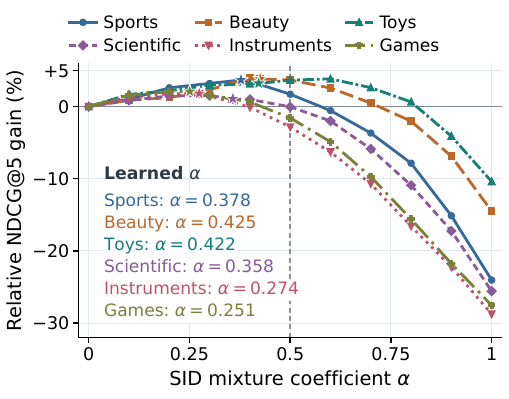}
\captionsetup{skip=5pt}
\caption{Validation NDCG@5 gains over item-only scoring. Stars: learned weights; dashed line: equal fusion.}
\Description{Six curves share one coordinate system, showing the percentage change in validation NDCG at five relative to item-only scoring as the SID coefficient varies from zero to one. Dataset colors, line styles, and markers distinguish the curves. Stars mark trained coefficients, and a vertical dashed line marks equal fusion. A color-matched list in the lower left records each dataset's learned coefficient.}
\label{fig:fusion-analysis}
\end{figure}

\subsection{Ablation Study}
\label{sec:ablation}

Table~\ref{tab:ablation} separates the roles of fused supervision, refinement transfer, and target construction. The interventions vary the fusion gate, SID representation, objectives, candidate support, and schedule; Appendix~\ref{app:ablation-protocol} defines each configuration.

\paragraph{\textbf{Supervising the served distribution.}}
OneChord improves NDCG@5 over the variant without fused likelihood on \textbf{all six datasets}, by \textbf{5.20--14.75\%}. Direct supervision of $q_\Theta$ therefore complements the branch objectives. Learning the gate also improves NDCG@5 over fixed equal fusion on all six domains, reaching \textbf{3.57\%} on Instruments. The single-component control varies the SID output representation while retaining the shared decoder; both component settings preserve one-pass prediction.

The fused-likelihood intervention preserves branch supervision and refinement objectives, isolating direct label supervision of the served distribution within the configured training objective. The fixed-gate intervention keeps this fused supervision and fixes the branch balance. Their results support the two linked operations in Section~\ref{sec:method-normalization}: supervise the combined probability and calibrate its mixing weight. The responsibility-weighted gradient in Eq.~\eqref{eq:branch-responsibility-gradient} explains how that supervision reaches the shared representations.

\paragraph{\textbf{Transferring refined rankings.}}
With refinement distillation, OneChord's NDCG@5 is \textbf{4.41--7.41\%} higher than the variant without any distillation across the six datasets. Removing SID distillation gives gaps of \textbf{2.47\%} on Beauty and \textbf{2.63\%} on Toys; removing ranking distillation gives \textbf{3.87\%} on Sports and \textbf{3.17\%} on Scientific. The choice of ranking target also matters: distilling the fused distribution improves NDCG@5 over item-only ranking distillation on \textbf{every domain}, by \textbf{1.40--3.24\%}. That control retains fused supervision and inference, linking the observed gains to the distribution used for transfer.

\paragraph{\textbf{Target support and candidates.}}
Catalog-normalized SID targets improve NDCG@5 over unrestricted targets by $2.79\%$ on Sports and $2.47\%$ on Beauty. Shared hard candidates give corresponding gains of $4.24\%$ and $2.26\%$ over sampled candidates alone. These results support aligning refinement on legal candidates that compete for the final list. The scheduling control complements the objective and support interventions; all configurations retain one-pass inference.

\subsection{Adaptive Fusion Analysis}
\label{sec:fusion-analysis}

Figure~\ref{fig:fusion-analysis} varies $\alpha$ with both trained branches fixed and reports validation NDCG@5 relative to item-only scoring. The learned fusion outperforms both single-branch endpoints on \textbf{all six datasets}, gaining \textbf{1.08--3.83\%} over item-only scoring. Learned SID weights span $0.251$--$0.425$, with NDCG@5 within \textbf{0.61\%} of the highest value in each coefficient sweep. The branches thus supply complementary ranking evidence, and supervised gate learning selects an effective domain-level balance. Appendix~\ref{app:fusion-analysis} specifies the sweep.

The learned coefficients give both branches positive weight. Their variation across domains shows how one learned scalar accommodates different evidence balances. Because the sweep holds the trained predictors fixed, it isolates score combination from representation changes; the fixed-fusion ablation separately examines their joint training. Together, these views connect the learned operating point to both the scoring function and the training objective.

\subsection{Performance across Interaction Sparsity}
\label{sec:sparsity-analysis}

We compare OneChord with the previous state-of-the-art method, SPRINT~\cite{cai2026sprint}, across target-item popularity and effective user-history length. Table~\ref{tab:sparsity} reports full-catalog test performance averaged equally over the six datasets.

Groups are defined from the training data and the prediction context. Let $f(i)$ count an item's training interactions, and let $q_{50},q_{80}$ be the catalog's training-frequency percentiles. Tail, middle, and head targets satisfy $f(i)\leq q_{50}$, $q_{50}<f(i)\leq q_{80}$, and $f(i)>q_{80}$, respectively; equal-frequency items stay together. Effective histories contain the non-padding items visible to the predictor: short, medium, and long groups have at most $5$, $6$--$10$, and $11$--$20$ interactions. Popularity uses original item identities, and recommendation accuracy uses the same exact-SID criterion as the main evaluation. Appendix~\ref{app:sparsity-analysis} gives the domain thresholds and aggregation procedure.

% Generated by paper/tables/make_sparsity_table.py from frozen test results.
\begin{table}[t]
\caption{Full-catalog test performance by interaction sparsity, averaged equally over six datasets. Gains are relative to SPRINT.}
\label{tab:sparsity}
\centering
\begingroup
\small
\renewcommand{\arraystretch}{1.12}
\OCMainTableFit{7}{l*{6}{r}}{%
\toprule
 & \multicolumn{2}{c}{SPRINT} & \multicolumn{2}{c}{\textbf{OneChord}} & \multicolumn{2}{c}{Gain} \\
\cmidrule(lr){2-3}\cmidrule(lr){4-5}\cmidrule(lr){6-7}
Group & R@5 & N@5 & R@5 & N@5 & R@5 & N@5 \\
\midrule
\rowcolor{OCTableHeader}
\multicolumn{7}{l}{\textit{Target-item popularity}} \\
Tail & 0.0263 & 0.0191 & \textbf{0.0276} & \textbf{0.0198} & +4.88\% & +3.62\% \\
Middle & 0.0355 & 0.0249 & \textbf{0.0408} & \textbf{0.0293} & +14.92\% & +17.70\% \\
Head & 0.0800 & 0.0527 & \textbf{0.0851} & \textbf{0.0569} & +6.30\% & +8.01\% \\
\midrule
\rowcolor{OCTableHeader}
\multicolumn{7}{l}{\textit{Effective history length}} \\
Short & 0.0581 & 0.0394 & \textbf{0.0616} & \textbf{0.0419} & +6.02\% & +6.60\% \\
Medium & 0.0524 & 0.0356 & \textbf{0.0570} & \textbf{0.0389} & +8.71\% & +9.06\% \\
Long & 0.0505 & 0.0334 & \textbf{0.0545} & \textbf{0.0377} & +7.86\% & +12.82\% \\
\bottomrule
}
\endgroup
\end{table}

OneChord improves both metrics in every averaged group. The largest item-group gains occur for middle-frequency targets: \textbf{14.92\%} in Recall@5 and \textbf{17.70\%} in NDCG@5. Tail items also benefit, with gains of $4.88\%$ and $3.62\%$, respectively, extending the improvements to less frequently observed targets. Across user groups, NDCG@5 gains rise from $6.60\%$ for short histories to $9.06\%$ for medium histories and \textbf{12.82\%} for long histories. The strongest average ranking improvement thus occurs when more historical interactions are available.

Each partition covers the complete test population and retains complete-catalog competition. Gains across both partitions show that the aggregate improvement extends across target frequencies and user-context lengths. The larger relative NDCG@5 gains for middle-frequency targets and long histories identify cohorts where the learned fused ranking is especially effective.

\section{Conclusion}
\label{sec:conclusion}

OneChord unifies next-item supervision, refinement distillation, and inference around \textbf{the served ranking distribution}. Learned fusion combines joint SID evidence and item matching; teacher-assisted states let the student construct refined fused targets. Catalog normalization and shared hard candidates preserve the scoring rule during transfer. Our analysis characterizes candidate-conditioned fusion and establishes a margin-based condition for full-catalog Top-$K$ agreement.

Across six domains, OneChord improves Recall@5 by 6.30\% and NDCG@5 by 7.12\% on average over the strongest baseline, with 1.04--1.18$\times$ faster generation than the fastest baseline. Gains extend to every averaged popularity and history group. One history encoding and one parallel decoding carry \textbf{distill what you rank} to accurate, efficient recommendation on the Web.

\bibliographystyle{ACM-Reference-Format}
\bibliography{references}

\appendix
\section{Derivations and Method Details}
\label{app:method}

\subsection{Teacher Kernel and Legal Completions}
\label{app:teacher-kernel}

The position kernel is specified in Eq.~\eqref{eq:explicit-kernel}. Let $\mathcal S(x)=\{s\in\mathcal S:s_\ell=x_\ell\text{ for every visible }\ell\}$. A state $x_t\sim Q_t(\cdot\mid s^+)$ has $s^+\in\mathcal S(x_t)$. Visible-code copying restricts the teacher to this nonempty set. Each sampled completion $\widehat s$ agrees with the source and supplies every newly revealed code, giving $\widehat s\in\mathcal S(x_u)$. Induction proves legal support for repeated transitions. Finite logits give positive mass to compatible SIDs and positive normalizers.

\subsection{Stopped-Target Composition}
\label{app:composition}

\paragraph{\textbf{Refinement average.}}
Fix $z_t$, $u<t$, and stopped predictions $a_x(s)=\bar p_\theta(s\mid H,x,u)$. On the finite refinement space,
\begin{equation}
m(s)=\sum_{x\in\mathcal X}K_\psi(x\mid z_t;u)a_x(s).
\end{equation}
Nonnegativity and $\sum_sm(s)=1$ follow from the normalized kernel and predictions. For $b(s)=p_\theta(s\mid z_t)$, interchanging finite sums gives
\begin{equation}
\begin{aligned}
\mathbb E_x\CE(a_x,b)
&=-\sum_xK_\psi(x\mid z_t;u)\sum_sa_x(s)\log b(s)\\
&=-\sum_sm(s)\log b(s)=\CE(m,b).
\end{aligned}
\label{eq:composition-barycenter}
\end{equation}
This proves Eq.~\eqref{eq:sampled-composition}. Writing $\log(a_x/b)=\log(a_x/m)+\log(m/b)$ and averaging proves Eq.~\eqref{eq:refinement-kl-decomposition}. Since $\CE(m,b)=\mathcal H(m)+\KL(m\|b)$, the optimum over unrestricted online distributions is $b=m$. With stopped targets and transitions, the refinement-variation term is constant for that update. The stopped-target gradients satisfy
\begin{equation}
\nabla_\theta\CE(\sg[a],b_\theta)
=\nabla_\theta\KL(\sg[a]\|b_\theta).
\end{equation}

\paragraph{\textbf{Finite noise measure.}}
Let $\delta_a$ be a unit point mass at $a$ and $\rho$ the probability of forcing a fully masked source. Composition samples from
\begin{equation}
\nu_D(\mathrm dt)=\rho\,\delta_1(\mathrm dt)
+\frac{1-\rho}{D-1}\sum_{j=2}^{D}\delta_{j/D}(\mathrm dt).
\label{eq:discrete-noise-measure}
\end{equation}
Thus $\int f(t)\nu_D(\mathrm dt)=\rho f(1)+(1-\rho)(D-1)^{-1}\sum_{j=2}^Df(j/D)$. With $D=4$ and $\rho=1/2$, the masses at $1/2$, $3/4$, and $1$ are $1/6$, $1/6$, and $2/3$. Substituting this measure into Eq.~\eqref{eq:composition} gives the expected sampled loss, including $1/L$. Ranking consistency uses the full-source case $t=1$ and forms candidates for each sampled refinement.

\subsection{Ranking Fidelity and Distribution Mass}
\label{app:ranking-proof}

Proposition~\ref{prop:ranking-fidelity} and its proof concern candidate agreement with Top-$K$ coverage. To characterize distribution mass, let $\widehat P$ and $\widehat Q$ be the temperature-scaled full-catalog refined and direct predictions. Write $a=\widehat P(\mathcal A)$ and $b=\widehat Q(\mathcal A)$. For $0<a,b<1$, the KL chain rule gives
\begin{equation}
\begin{aligned}
\KL(\widehat P\|\widehat Q)
&=\operatorname{kl}(a\|b)+a\KL(P\|Q)\\
&\quad+(1-a)\KL(\widehat P_{\mathcal A^c}\|
\widehat Q_{\mathcal A^c}),
\end{aligned}
\end{equation}
where $\operatorname{kl}$ is Bernoulli KL. The candidate objective controls the within-set term; Top-$K$ coverage and the margin give the selection guarantee. The remaining terms characterize set-mass and outside-candidate agreement.

\subsection{Ranking-Consistency Gate Gradient}
\label{app:gate-proof}

For fixed branch predictions, $\partial\log q_1(s)/\partial\beta=\rho_s-\alpha$. Set $\ell_s=\log q_1(s)/T$ on the fixed candidate set; then $Q=\operatorname{softmax}(\ell)$. With $P$ stopped, differentiating $T^2\KL(P\|Q)$ gives
\begin{equation}
\begin{aligned}
\frac{\partial\,[T^2\KL(P\|Q)]}{\partial\beta}
&=T\sum_{s\in\mathcal A}(Q(s)-P(s))(\rho_s-\alpha)\\
&=T\sum_{s\in\mathcal A}(Q(s)-P(s))\rho_s.
\end{aligned}
\end{equation}
The $\alpha$ term cancels because $P$ and $Q$ each sum to one. This proves Eq.~\eqref{eq:ranking-gate-gradient}, including the effect of post-fusion restriction and temperature.

\subsection{Architecture and Training Specification}
\label{app:implementation}

\paragraph{\textbf{Index and backbone.}}
The fixed four-code OPQ index~\cite{ge2014opq} has $256$ codes per position. Identical SIDs share one catalog entry. The teacher and student share this mapping and catalog. The backbone has hidden width $256$, four attention heads, one history-encoder layer, four decoder layers, feed-forward width $1024$, and dropout $0.1$. Histories retain $20$ items. The decoder uses position-specific mask embeddings and a sinusoidal time embedding passed through an MLP. Padding is masked in history attention and decoder cross-attention.

\paragraph{\textbf{Prediction heads.}}
The student uses $R=4$ components with the parameterization in Eq.~\eqref{eq:component-parameterization}; $B_\theta$ is a shared two-layer GELU MLP. The teacher has one product-categorical component and the same backbone dimensions, with token logits directly from decoder states. The shared item projection $M_\theta$ uses a two-layer ReLU MLP. We use $\tau_{\mathrm{item}}=0.1$ on Instruments and $0.01$ elsewhere, and initialize $\beta=0$.

\paragraph{\textbf{Noise and targets.}}
With $D=4$, teacher training samples $t$ uniformly from $\{1/D,\ldots,1\}$ and independently replaces it with $1$ with probability $0.25$. Student composition uses Eq.~\eqref{eq:discrete-noise-measure}. Ranking KL is averaged over full-source minibatch examples and is zero when none are present; its online prediction reuses the direct supervised pass. Stopped targets use the current student with detached history memory and dropout disabled.

\paragraph{\textbf{Candidates and coefficients.}}
The supervised multiset contains the positive SID and $100$ uniform negatives sampled with replacement, excluding the positive. Ranking distillation adds $b=32$ candidates per predictor and uses $T=2$. We set $\lambda_s=\lambda_i=\lambda_f=\lambda_b=\lambda_c=1$ and $\lambda_r=0.1$. With $[a]_0^1=\min(1,\max(0,a))$,
\begin{equation}
\begin{aligned}
w(n)&=\left[\frac{n-n_0}{n_{\mathrm{ramp}}}\right]_0^1,\\
d(n)&=1-(1-\eta)
\left[\frac{n-n_{\mathrm{start}}}{n_{\mathrm{end}}-n_{\mathrm{start}}}\right]_0^1.
\end{aligned}
\end{equation}
Step counts correspond to $5$ supervised epochs, $10$ ramp epochs, and composition decay from epoch $35$ to $75$, with $\eta=0.25$. Fused supervision is active throughout; boundary and ranking consistency retain their post-ramp coefficients during composition decay. Chunked catalog sums preserve exact normalization.

\section{Experimental Protocol}
\label{app:experimental-protocol}

\subsection{Datasets and Evaluation}
\label{app:datasets}

\begin{table}[t]
\caption{Dataset statistics after preprocessing, excluding padding. SID counts refer to distinct four-code identifiers; Avg.\ is the mean sequence length.}
\label{tab:dataset-statistics}
\centering
\small
\setlength{\tabcolsep}{3.2pt}
\begin{tabular}{@{}lrrrrr@{}}
\toprule
Dataset & Users & Items & SIDs & Interactions & Avg.\\
\midrule
Sports & 35,598 & 18,357 & 15,102 & 296,337 & 8.32\\
Beauty & 22,363 & 12,101 & 10,377 & 198,502 & 8.88\\
Toys & 19,412 & 11,924 & 10,220 & 167,597 & 8.63\\
Scientific & 50,985 & 25,848 & 19,810 & 412,947 & 8.10\\
Instruments & 57,439 & 24,587 & 20,951 & 511,836 & 8.91\\
Games & 94,762 & 25,612 & 20,944 & 814,586 & 8.60\\
\bottomrule
\end{tabular}
\end{table}

For Amazon Reviews 2014, we use the public category-specific 5-core interactions. For Amazon Reviews 2023, we use the public \texttt{5core/last\_out\_w\_his} release and recover each user's longest chronological sequence from its supplied histories and targets. The last two interactions are held out for validation and testing. Training examples slide over the preceding sequence, using up to $20$ historical items for each next-item target. Validation and test histories contain the interactions preceding their respective targets. Validation and testing each evaluate one held-out target per user.

The fixed OPQ mapping covers every processed item. Evaluation ranks all distinct legal SIDs and scores exact four-code matches to held-out targets. Identical SIDs share one catalog entry. The mapping, catalog, and target labels are shared for the SID-based comparison.

\subsection{Baseline Protocol}
\label{app:baseline-protocol}

Table~\ref{tab:main-results} uses common chronological splits and full-catalog evaluation. The fastest baseline and OneChord share the OPQ mapping, catalog, and evaluator in the efficiency comparison.

\subsection{Optimization and Reproducibility}
\label{app:optimization-timing}

Appendix~\ref{app:implementation} specifies the backbone, prediction heads, noise grid, candidate construction, loss coefficients, and schedules.

\paragraph{\textbf{Optimization.}}
The student and the fastest baseline use AdamW with learning rate $10^{-4}$, zero weight decay, gradient clipping at $1$, and batch size $1024$. A cosine schedule uses $10{,}000$ warmup steps. Training lasts up to $150$ epochs, with full-validation evaluation each epoch and patience $10$. Denoisers train for up to $100$ epochs with the same backbone dimensions and $10\%$ warmup. Students have independent initialization. Distillation variants within a dataset share the frozen denoiser, mapping, and splits.

\paragraph{\textbf{Random seeds.}}
Main accuracy results use five random seeds: $2024$, $2025$, $2026$, $2027$, and $2028$. Each run selects its checkpoint using the complete validation set. The selection score is $0.8\,\mathrm{N@10}+0.2\,\mathrm{R@10}$. The selected checkpoint is then evaluated on the complete test set. Validation and test batch sizes are both $256$.

\paragraph{\textbf{Accuracy reporting.}}
Table~\ref{tab:main-results} reports mean test performance; each OneChord cell shows the mean with its standard deviation across the five runs directly below. Each relative improvement is computed as $(m-b)/b\times100\%$, where $m$ is OneChord's mean and $b$ is the strongest baseline mean in that column. The reported percentages use the displayed means. Cross-domain improvements are arithmetic means of the six per-domain relative gains, yielding $6.30\%$ for Recall@5 and $7.12\%$ for NDCG@5.

\subsection{Ablation Protocol and Definitions}
\label{app:ablation-protocol}

The twelve ablated configurations use seed $2024$ over all six datasets, sharing the OPQ mapping, chronological splits, catalog, and backbone dimensions. Distillation controls share a frozen denoiser within each dataset. Checkpoint selection uses the complete-validation criterion above, followed by complete-test evaluation. Section~\ref{sec:ablation} computes relative differences from the displayed Recall@5 and NDCG@5 scores.

\paragraph{\textbf{Fusion and SID representation.}}
\emph{Fixed equal fusion} holds $\alpha=1/2$ throughout training and inference. \emph{Without fused likelihood} sets $\lambda_f=0$, retaining the branch objectives and distillation. \emph{Single SID component} sets $R=1$ and retains the first initialized component; the backbone and initial data order match those of the full model.

\paragraph{\textbf{Distillation objectives.}}
\emph{Without all distillation} sets $\lambda_b=\lambda_c=\lambda_r=0$. \emph{Without SID distillation} sets $\lambda_b=\lambda_c=0$. The boundary, composition, and ranking controls individually set $\lambda_b$, $\lambda_c$, and $\lambda_r$ to zero. \emph{Item-only ranking distillation} constructs targets and hard candidates from $r_\theta$, retaining $q_\Theta$ for supervised fused likelihood and inference. This intervention isolates the distillation distribution within the fused recommender.

\paragraph{\textbf{Target support and scheduling.}}
\emph{Without catalog normalization} applies boundary and composition distillation to the joint SID mixture without restricting its normalization to legal catalog SIDs. Its fused ranking objective remains unchanged. \emph{Without hard candidates} removes the additional top-scoring candidates from both predictions, retaining the labeled item and sampled supervised negatives. \emph{Without curriculum} sets the distillation ramp and composition-decay factors to one throughout training, retaining their configured loss weights. Every control retains one parallel decoding pass at inference.

\begin{figure*}[t]
\centering
\includegraphics[width=\textwidth,height=448pt,keepaspectratio]{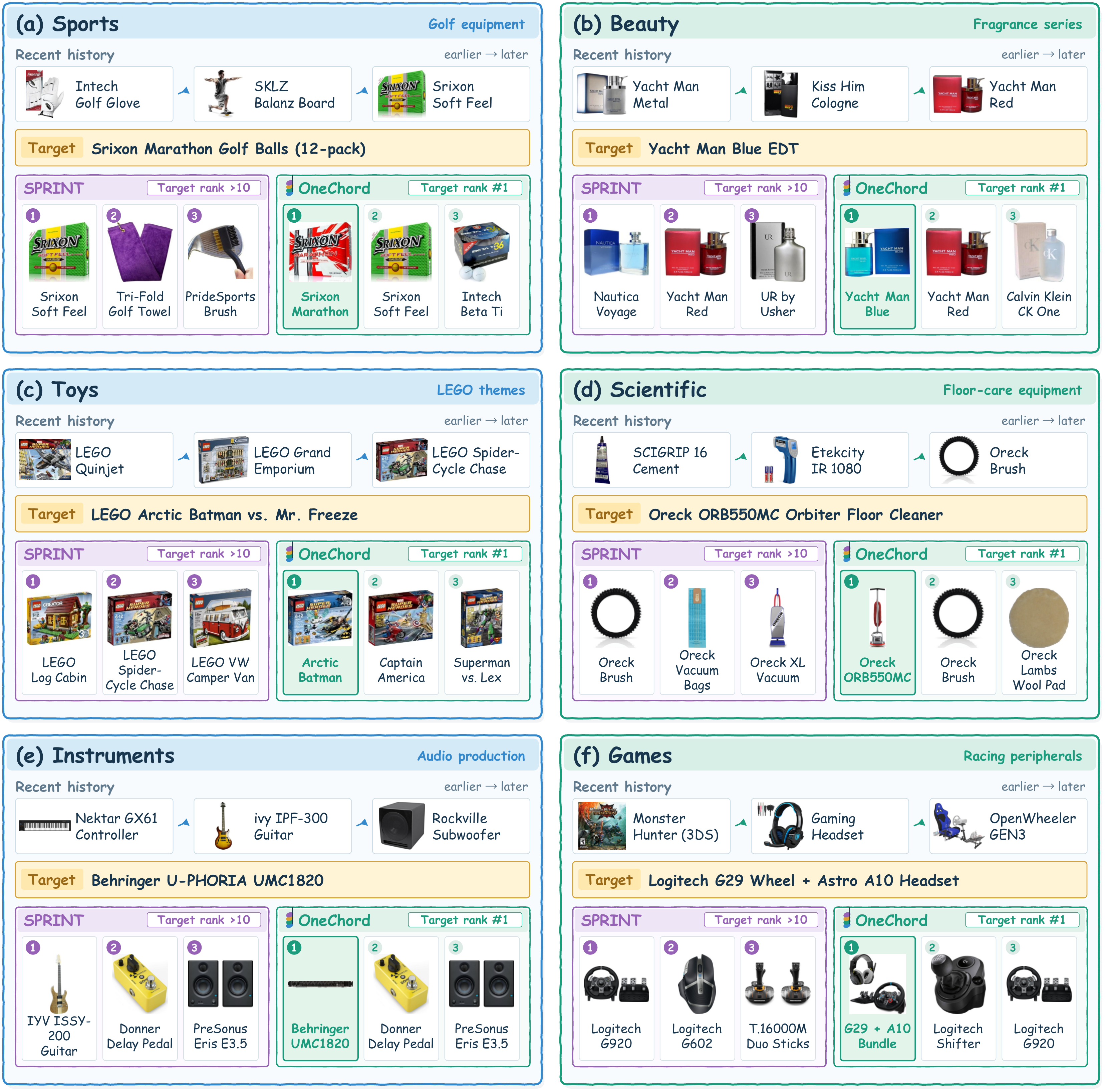}
\caption{Product images illustrate recent histories and original Top-3 outputs in six selected successful test cases. Green marks the target; rank denotes its Top-10 position. Shown targets and predictions have unique SIDs.}
\Description{Six recommendation cases comparing SPRINT and OneChord on Sports, Beauty, Toys, Scientific, Instruments, and Games. Each panel contains recent history, a held-out target, and two ranked recommendation lists. OneChord ranks the target first in all six selected cases; SPRINT does not retrieve it in its first ten results.}
\label{fig:case-examples}
\end{figure*}

\subsection{Fusion Sensitivity Protocol}
\label{app:fusion-analysis}

Figure~\ref{fig:fusion-analysis} uses the selected seed-$2024$ OneChord checkpoint for each dataset and evaluates all validation users over the same legal SID catalog. We hold all trained parameters fixed and vary the inference coefficient over $\{0,0.1,\ldots,1\}$, additionally evaluating the coefficient learned during training. The endpoints $\alpha=0$ and $\alpha=1$ use item evidence and SID evidence, respectively. They isolate each branch's scoring contribution within the same trained model. The plotted ordinate is $100[\mathrm{N@5}(\alpha)/\mathrm{N@5}(0)-1]$, placing each dataset's item-only score at zero. The fixed equal-fusion control in Table~\ref{tab:ablation} instead trains with $\alpha=1/2$ throughout. Stars identify the coefficients learned by the original training objectives; the sweep serves as a sensitivity analysis. All relative NDCG@5 differences use unrounded scores.

\subsection{Interaction Sparsity Protocol}
\label{app:sparsity-analysis}

Table~\ref{tab:sparsity} uses both methods' validation-selected seed-$2024$ checkpoints and the groups in Section~\ref{sec:sparsity-analysis}. Frequencies count chronological training interactions once, excluding held-out events and sliding-window repetitions. Quantiles round upward to observed frequencies: $(q_{50},q_{80})=(7,15)$ for Sports and Beauty, $(6,14)$ for Toys, $(7,14)$ for Scientific, $(8,17)$ for Instruments, and $(9,25)$ for Games.

The popularity and history partitions each cover all test users. Both models rank the full legal SID catalog before group scoring. Each group's Recall@5 and NDCG@5 are averaged equally over six datasets; relative gains use ratios of the unrounded macro means.

\subsection{Qualitative Recommendation Examples}
\label{app:case-examples}

Figure~\ref{fig:case-examples} illustrates one successful case per dataset using the checkpoints from Appendix~\ref{app:sparsity-analysis}. Both models preserve the original histories and rankings. OneChord retrieves each shown target at rank~1; SPRINT's Top-10 omits it. Product images follow the dataset metadata.

\end{document}